\documentclass[article,onecolumn]{IEEEtran}
\usepackage[dvips]{graphicx}
\usepackage{times}
\usepackage{amsthm}
\usepackage{graphicx}
\usepackage{algorithm}
\usepackage{algorithmic}
\newtheorem*{definition}{Definition}
\newtheorem{remark}{Remark}
\usepackage{caption}
\usepackage{subfigure}
\usepackage{url}
\usepackage{amsmath}
\usepackage{amsfonts}
\usepackage{amssymb}
\usepackage{mathrsfs}

\newtheorem{theorem}{Theorem}
\newtheorem{Lemma}[theorem]{Lemma}
\newtheorem{Proposition}[theorem]{Proposition}

\newtheorem{example}[theorem]{Example}
\usepackage{color, colortbl}
\definecolor{MRed}{rgb}{1,.6,.6}

\title{On Optimal Family of Codes for Archival DNA Storage}

\author{
\IEEEauthorblockN{Dixita Limbachiya, Vijay Dhameliya, Madhav Khakhar and Manish K. Gupta}
\IEEEauthorblockA{\\
Dhirubhai Ambani Institute of Information and Communication Technology\\
Gandhinagar, Gujarat, 382007 India\\
Email: dlimbachiya@acm.org and mankg@computer.org\\
}
}

\begin{document}
\maketitle
\IEEEpeerreviewmaketitle

\begin{abstract}
\footnote{
Dixita, Vijay and Madhav contributed equally. Corresponding author: Manish K. Gupta, email: mankg@computer.org.  Recently a re-writable DNA based storage system has been reported in \cite{rewritableDNAoligica}.
}
DNA based storage systems received attention by many researchers. This includes archival and re-writable random access DNA based storage systems.
In this work, we have developed an efficient technique to encode the data into DNA sequence by using non-linear families of ternary codes. In particular, we proposes an algorithm to encode data into DNA with high information storage density and better error correction using a sub code of Golay code. Theoretically, 115 exabytes (EB) data can be stored in one gram of DNA by our method.
\end{abstract}


\section{Introduction}
Big data explosion due to advancement in social networking and IoT (Internet of Things) produces immense data, which urge the data scientists to strive for the development of better data storage medium. Optical, digital and cloud data storage \cite{dimakis2010network} medium have their own limitations and need to be maintained regularly. While the computer scientists are endeavoring to develop dense data storage medium, researchers at the other end thought of exploring natural medium to safeguard this data.
Properties of DNA like scalability, density and long term stability makes it ideal for long term archival of data. A good deal of work has been done to store the data on DNA \cite{davis1996microvenus} \cite{cox2001long} \cite{wong2003organic} \cite{church2012next}. 
By employing error correction for DNA data storage, the breakthrough was made by N. Goldman and his team  \cite{goldman2013towards} in the year $2013.$ They developed very efficient and novel approach for storing the data of size $739$ KB on DNA and retrieving it back. They used four folds redundancy to tackle errors by retrieving correct data from one of the copies of DNA strands. But due to redundancy there was increase in the length of DNA which makes the technique expensive to use DNA for storage at the commercial level. In the year $2014,$ group at The Chinese University of Hong Kong developed method for larger DNA reads used to embed the data \cite{yim2014essential}.  Low Density Parity-Check (LDPC) codes based method was introduced by Aldrin Kay-Yuen Yim et al.,  to encode data in DNA for improved length of DNA reads. Most of the methods used so far have limitations in size of the data inserted in DNA and cost associated with DNA synthesis and sequencing technology.  Work of Goldman \cite{goldman2013towards} and Church \cite{church2012next} has laid a cornerstone for DNA based data storage systems. Recently model of DNA channel for data storage is proposed by Han Mao Kiah et al., \cite{KiahPM14}. Also a rewritable random access DNA based storage system has been very recently reported \cite{rewritableDNAoligica}. In \cite{grass2015robust}, Reed Solomon codes based DNA storage system has been reported. In this work, we have developed error correction scheme (a family of non linear ternary codes) by modifying Goldman's scheme \cite{goldman2013towards} for the error detection and correction along with improvement in storage capacity to store data on DNA. 

The paper is organized as follows. Section 2 describes briefly related work and the Golddman's approach. Section 3 gives our work on non-linear ternary codes. Section 4 discuss the analysis of class DNA Golay code sub code. Final section 5 concludes the paper with general remarks. 
\section{Error correction for DNA storage}
There are various error correction schemes used for DNA based data storage systems \cite{smith2003some} \cite{yachie2007alignment} \cite{haughton2011repetition} but the effective schemes was first introduce by Church \cite{church2012next} and Goldman \cite{goldman2013towards}.
Using next generation synthesis and sequencing technology, Church came up with efficient one bit per base algorithm  of encoding information bits into fix length of DNA chunks (99 bases). Flanking primers at the beginning and end of information data was inserted to identify the specific DNA segment in which the particular data was encoded.
This method includes homo polymer repeated sequences that cause errors while writing and reading the data. This difficulty was taken care by Goldman's approach for DNA based data storage in $2013.$ Goldman used one bit per base system introduced by Church and modifying it by employing the improved base $3$ Huffman code (trits 0, 1 and 2) encoding scheme.
In this original file type to binary code which is then converted to a ternary code  which is in turn converted to the triplet DNA code. It involved four steps shown in Fig \ref{AdvancedDNASchemetic}. Binary digits holding the ASCII codes was converted to base-3 Huffman code that replaces each byte with five or six base-3 digits (trits). Each of trit was encoded with one of the three nucleotides different from the previous one used (to avoid homopolymers that cause error in the synthesis of DNA). DNA strand was divided into chunks each of length $117$ base pair (bp). To include redundancy for error detection and correction, $75$ bases for each DNA information chunks were overlapped for  four fold redundancy to get the data loss that occurred during synthesis and sequencing DNA. For the data security each redundant chunk was converted to reverse complement of the strand in every alternate chunks. Each DNA chunk was appended with an indexing information bit to determine the location of the data in each chunk and the file from which it is generated. At the end parity check bit was added for error detection. For details the reader is referred to \cite{goldman2013towards}. We modified the Goldman's scheme further as described in next section. 
%

\begin{figure}[h]
  \begin{center}
    \scalebox{0.32}[0.35]{\includegraphics{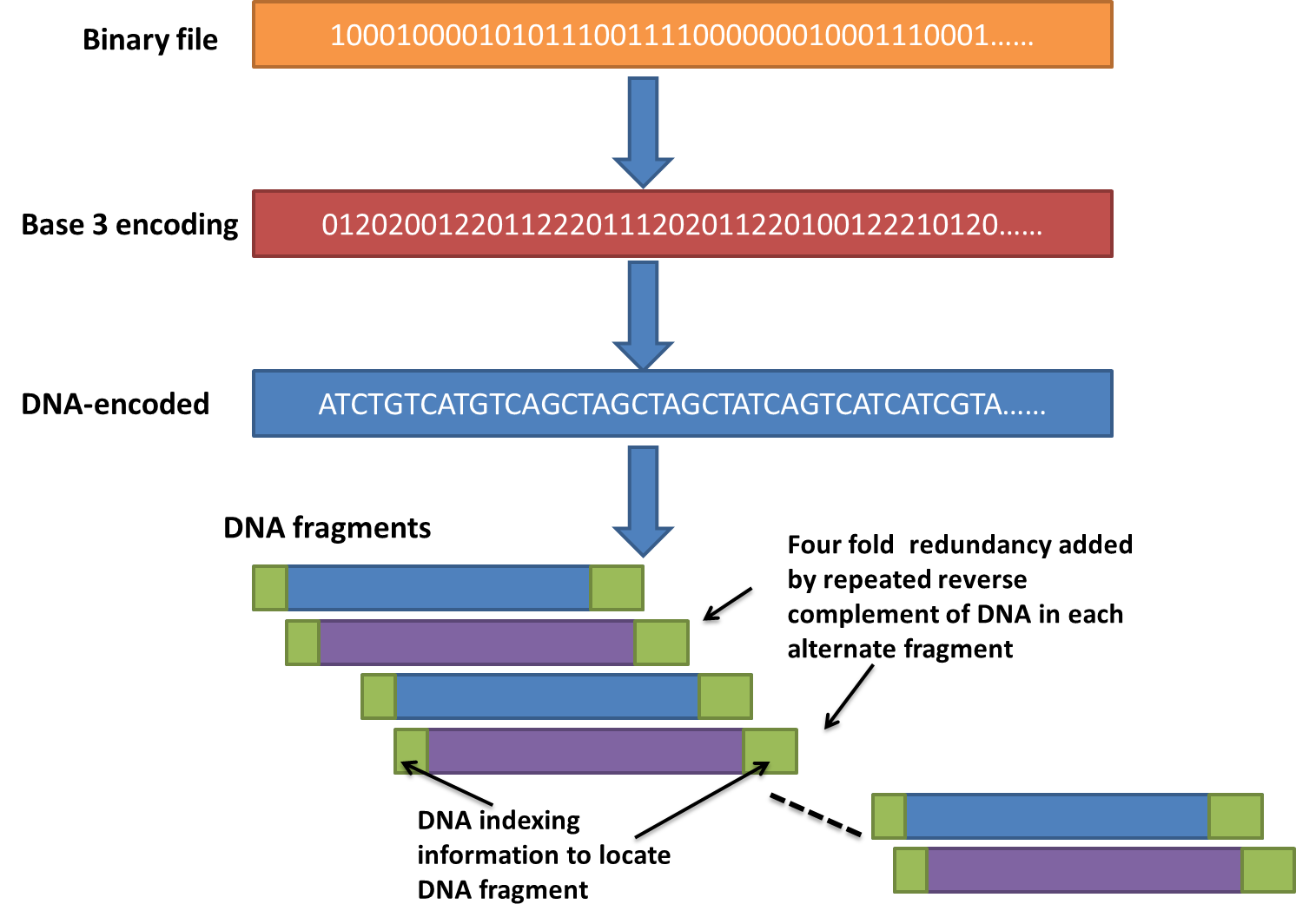}}
    \caption{Stepwise encoding of data into DNA using Goldman's approach is explained in detail. In this original file type to binary code (0, 1) which is then converted to a ternary code (0, 1, 2), which is in turn converted to the DNA code. Four fold redundancy is added to each DNA chunk by introducing reverse complement to every alternate DNA chunk. DNA chunk index is addedto each DNA chunk to locate the file.}
    \label{AdvancedDNASchemetic}
  \end{center}
\end{figure}
\section{Family of Non-Linear Ternary Codes}
Goldman used four fold redundancy due to which was increase in the length of the DNA used to encode data which makes the DNA based storage medium more expensive. In this work, non-linear ternary error correction codes is used instead of Huffman codes as referred in Fig \ref{golaycode}. In DNA we have only three nucleotides different from the last one. So it is inescapable to have codewords in base-3 in order to avoid homopolymers runs \cite{goldman2013towards}. The non-linear ternary codes described here improve the storage capacity of DNA by decreasing the length of the DNA required for storing a file. Any arbitrary computer file can be converted into list of ASCII values ranging from $0$ to $255$. So we need set of $256$ ternary codewords, each corresponding to one value in $\{0,..,255\},$ to encode any such file into DNA string. From exhaustive search and enumeration of codewords we identified and constructed seven families of non linear ternary codes with the parameters $(9, 256, 3)_{3}$, $(11, 256, 5)_{3}$, $(15, 256, 7)_{3}$, $(18, 256, 9)_{3}$, $(21, 256, 11)_{3}$, $(24, 256, 13)_{3}$ and $(26, 256, 15)_{3}$ respectiveley \cite{codetables}. Among these, $(11, 256, 5)_{3}$ is a sub code of ternary Golay code (listed in the Table \ref{DNA golay code})  \cite{2401740}. We discuss encoding using the subcode $(11,256,5)_{3}$ in detail. The ternary Golay code consists of $729$ codewords of length $11,$ with minimum hamming distance $5.$ So it allows receiver to identify 4 trits of errors and correct 2 trits of errors that occur in codeword.
The steps b and c of the algorithm shown in Fig \ref{AdvancedDNASchemetic} were modified by using ternary Golay codes instead of Huffman codes (see Fig \ref{DNAgolay}). Each byte of a computer file is encoded to Golay code Table \ref{DNA golay code}. The table consists of $243$ codewords from $729$ Golay codewords such that minimum hamming distance between any two codewords is $6.$ These $243$ codewords were assigned to $243$ ASCII values having highest probability of occurrences according to frequency table used by Goldman as described in supplementary file. Remaining $13$ ASCII values were assigned with codewords chosen randomly from remaining set of ternary Golay codewords i.e. these $13$ codewords will have minimum hamming distance $5$ with other $243$ codewords. This ensures that by maximum likelihood decoding we can correct up to $2$ trit flip error. Since there exists only $243$ codewords in set of $3^{11}$ codewords with minimum hamming distance $6,$ it is not possible to construct code consisting of minimum $256$ codewords such that length of codewords is $11$ and minimum hamming distance more than $5.$ The file information (i.e. file size and file extension) was encoded using same table and appended at the end of this string with Golay codes of comma and semicolon as separators. Next, this string of trits was converted to DNA such that only one of three nucleotides, different from the last one, was used to encode current bit into nucleotide. The resulting DNA sequence was split into segments (or chunks) of $99$ bases each as shown in Fig \ref{chunk} where $(i= 99$) is number of bases used for storing file content. Here, $\lambda$ is number of bases required to store file index number (no of file index trits = 2, thus allowing maximum of $9$ files to be distinguished), $\mu$ is number of bases required for storing chunk index (no of segment index trits $\mu   = \left \lceil \log_{3}{\mbox{( total no of segments)}} \right \rceil  )$  and one base for odd parity-check trit \cite{goldman2013towards} were appended at the end of each segment of $99$ bases. This parity is obtained by summing odd bits for file identifier and chunk index.
 \begin{figure}[h]
  \begin{center}
    \scalebox{0.4}[0.4]{\includegraphics{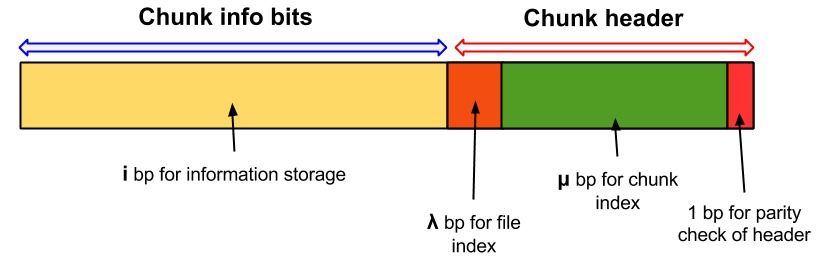}}
    \caption{Chunk architecture for DNA Golay code approach. There is two main parts of chunk, one is chunk info bits and other is chunk header. Chunk info bits contains original data information and chunk header has index. Chunk header includes file index for identification of file and chunk index to identify chunk. Last is appended odd parity check information.}
    \label{chunk}
  \end{center}
\end{figure}
%
 \begin{figure}[h]
  \begin{center}
    \scalebox{0.32}[0.32]{\includegraphics{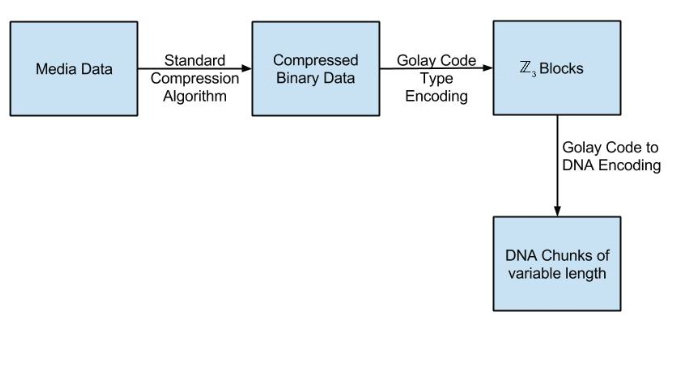}}
    \caption{Schematic flow diagram for DNA Golay code approach. Computer file is compressed using standard compression methods. Compressed binary file is converted to non linear family of codes by assigning each character to codeword. This ternary code is converted to DNA and DNA sequence is fragmented into chunk of variable length. No redundancy is added to chunk. Error correction can be done by nearest-neighbour likelihood decoding approach.}
    \label{golaycode}
  \end{center}
\end{figure}
\begin{figure}[h]
  \begin{center}
    \scalebox{0.22}[0.25]{\includegraphics{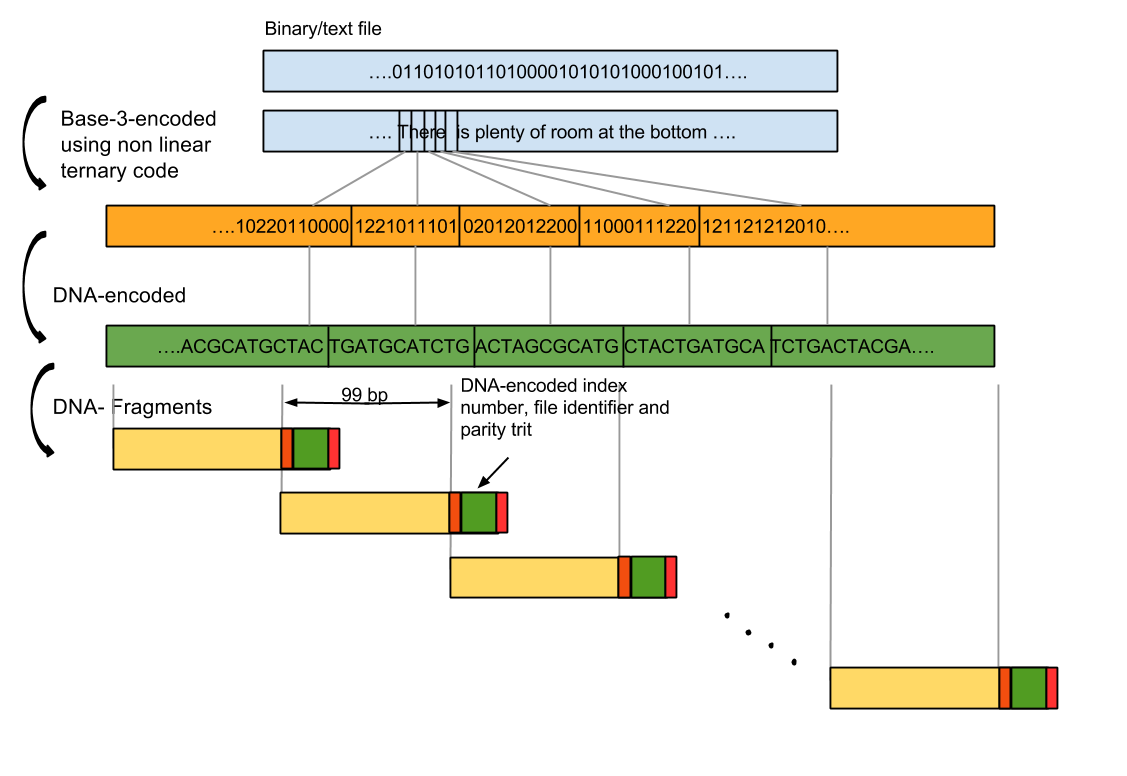}}
    \caption{Improved error correction approach for DNA data storage using DNA Golay code is developed here. It depicts the steps of conversion of file into DNA sequence by using DNA Golay code table. In first step, binary code (blue bar) was mapped to base 3 non linear ternary code (orange bar) where each byte was replaced with base 3 numbers called trits (for ternary Golay code with 11 trits). These trits were then converted to DNA code (green bar) by replacing each trit with one of the three nucleotides different from the previous one used to avoid homopolymers. Long DNA was divided into DNA chunks of length 99 base pairs without redundancy. Each DNA segment was appended with index, file identifier and parity trits.}
    \label{DNAgolay}
  \end{center}
\end{figure}
\begin{table}[ht]
\caption{Goldman base table designed to covert Huffman code to DNA nucleotides avoiding homo polymer}
\centering
\begin{tabular}{|l|l|l|l|}
\hline
&0&1&2 \\ \hline
A&C&G&T \\ \hline
C&G&T&A \\ \hline
G&T&A&C  \\ \hline
T&A&C&G \\ \hline
\end{tabular}
\label{basetable}
\end{table}

\begin{example}
To understand the encoding and decoding procedure, here is demonstrated simple example. Message data is "DA". Each letter of message is converted to ASCII characters. ASCII characters for each letter are 68 and 65 respectively. ASCII characters are converted to sub code of ternary Golay code of length 11. Codes corresponding to letter D is 02221221120 and for A is 10111000101 referred from table \ref{DNA golay code}. 
DNA sequence corresponding to message DNA is CATGATGCTGAGTCTCGTAGTC. This DNA sequence is divided to chunks $C_1$ is CATGATGCTGA and $C_2$ is GTCTCGTAGTC (here of length 11). Index of chunk i and file identifier (ID) is appended at the end of the message data. Let the chunk index for each DNA chunk be 0 and 1. File identifier for the file be 00. Also for checksum, parity bit P is added which can be calculated by summation of odd position trits in ID and i. Now each DNA chunk is appended by with chunk identifier bases (concatenating ID. i .P). Chunk identifier for each DNA chunk here is CGTA and CGAG respectively. So final DNA chunks are CATGATGCTGACGTA and GTCTCGTAGTCCGAG.
\end{example}
\subsection{Error Correction Analysis}
\begin{definition}
Let $\Sigma_{DNA} =\{A,T,G,C\}$ be the alphabet of DNA nucleotides.  A DNA code $\mathscr{C}_{DNA}(n,M)$ is  a sub set of $\Sigma_{DNA}^n$ of size $M$ with each codewords of length $n$.
\end{definition}
\begin{definition}
The Hamming distance $d_{H}(\bar{\textbf{x}}_{\textbf{DNA}}$,$\bar{\textbf{y}}_{\textbf{DNA}})$ between two DNA codewords $\bar{\textbf{x}}_{\textbf{DNA}}$ and $\bar{\textbf{y}}_{\textbf{DNA}}$ is the number of positions in which $\bar{\textbf{x}}_{\textbf{DNA}}$ and $\bar{\textbf{y}}_{\textbf{DNA}}$ differs. For instance, let $\bar{\textbf{x}}_{\textbf{DNA}}$ = ATGACT $\bar{\textbf{y}}_{\textbf{DNA}}$ = ATTAGC, then $d_{H}(\bar{\textbf{x}}_{\textbf{DNA}}$,$\bar{\textbf{y}}_{\textbf{DNA}})$ = 3.
\end{definition}
\begin{Lemma}\label{uniq}
Let $\phi$ be the map used for the conversion (using Table \ref{basetable}) of $\mathscr{C}_{Z_3}$ ternary code to $\mathscr{C}_{DNA}$ DNA code i.e, $\phi(\mathscr{C}_{Z_3})=\mathscr{C}_{DNA}$.
Observe that the map $\phi$ is not unique.
\label{map}
\end{Lemma}
\begin{proof} Easy to observe.
\end{proof}
\begin{Lemma}
 Let $\bar{\textbf{x}}_{\textbf{DNA}}$ $\in \mathscr{C}_{DNA}$ send through a noisy  channel and $\bar{\textbf{y}}_{\textbf{DNA}}$ is received. If 
 $d_{H}(\bar{\textbf{x}}_{\textbf{DNA}}$,$\bar{\textbf{y}}_{\textbf{DNA}})$ $=$ $t$  and corresponding ternary code obtained using Goldman's base Table \ref{basetable} for $\bar{\textbf{x}}_{\textbf{DNA}}$  and  $\bar{\textbf{y}}_{\textbf{DNA}}$ are  $\bar{\textbf{x}}$ and $\bar{\textbf{y}}$ respectively then $t <$ $d_{H}(\bar{\textbf{x}}$,$\bar{\textbf{y}})$ $\leq$ $2t$.
 \label{distance}
\end{Lemma}

\begin{proof} Using Lemma \ref{uniq}.
\end{proof}
\begin{example}
let $\bar{\textbf{x}}_{\textbf{DNA}} = GTCTCGTAGTC$ and $\bar{\textbf{y}}_{\textbf{DNA}} = GAGTCGTAGTC$ then $\bar{\textbf{x}} = 10111000101 $  and $\bar{\textbf{y}} = 11101000101 $. The minimum hamming distance $d_{H}(\bar{\textbf{x}}_{\textbf{DNA}}$,$\bar{\textbf{y}}_{\textbf{DNA}}) = 2$ and $d_{H}(\bar{\textbf{x}}$,$\bar{\textbf{y}})= 2$.
\end{example}
%
%
%
\begin{Lemma}
Using double layer error correcting scheme, we can correct any 2 bit flips in DNA.
\end{Lemma}
\begin{proof}
Correction of 1 bit flip in DNA is straightforward using Lemma \ref{uniq}. In case of 2 bit flips we can prove it by contradiction and Lemma 1 and 2.
\end{proof}
\begin{remark}
If error occurs in chunk header, error cannot be corrected and data cannot be retrieved.
\end{remark}
\section{Analysis}
Performance of error correcting codes used for DNA data storage can be measured by Shannon information capacity, information ratio of base to bits. 
Family of non linear codes described here has many advantages over codes used for DNA data storage. It has properties like higher data density, higher DNA Shannon information, lower DNA storage cost and code rate almost $0.5$. We describe them now briefly.  
\subsection{DNA Information Density}
Amount of data that can be encoded in DNA can be quantified by DNA information density. \begin{definition} For a given DNA based information storage system, DNA information density is total amount of data that can be stored in unit gram of DNA.\end{definition} At theoretical maximum, one gram of single stranded genetic code can store $455$ EB (exabytes) of information \cite{church2012next}. Goldman achieved information density $2.2$ PB (petabytes) per gram of DNA. Using DNA Golay codes, computationally we achieved information density for DNA based storage medium as $1.15 \times 10^{20} = 115$ EB (Exabytes) per gram DNA. 
\begin{Proposition}
DNA storage capacity using our chunk architecture for one gram of DNA is calculated by solving the following non-linear equation
$$
\begin{array}{c}
[(182 \times 10^{19} \times l) \div \left( \left( l+3\right)  + log_3 \left( \frac{N\left( x+22\right)}{l} \right) \times N \right)] -22\\
= x,
\end{array}
$$ where x= number of  bytes per one gram of DNA, l = Length of chunk without chunk index, N =  Length of the error correcting code.
\end{Proposition}
In case of DNA Golay codes, DNA storage capacity obtained is $1.15 \times 10^{20}$ (115 Exabytes) bytes per gram of DNA  by using $l = 99$ and $N =11$.
%
\subsection{DNA Storage Cost}
Cost of DNA synthesis and sequencing is a major limitation of DNA storage medium. Considering synthesis cost to be $\char36 0.05$ per base \cite{goldman2013towards}, total cost per MB for different file size (bytes) was  plotted by us. Using DNA Golay code, required amount of the DNA decreased which resulted in decreased cost. We also plotted the cost to store data using various non-linear ternary codes of our family.  We observe that with the increase in the amount of data, there is negligible increase in cost per unit data (1 MB) associated with DNA based storage using non linear family of codes developed by us. In contrary, cost for the DNA increases with the increase in data using Goldman method.
\subsection{Tradeoff for Code Rate }
The information rate of our encoding scheme is 0.73 bits per base by encoding 8 bits into 11 DNA bases per byte (i.e. 8/11 = 0.73 bits per base). It can be improved further to 0.89 by using 9 base per byte (i.e. 8/9 = 0.89 bits per base) by encoding each byte using code $(9, 256, 3)_{3}$, at the cost of reduction in error detection capacity to 2 trits and correction capacity to 1 trits per codeword. 
 It can be observed that error correction capacity t for each code increases linearly with increase in length of the code n. But with increase in these two parameters, there is declination in code rate. For reliable data storage, code rate is optimal with $(11, 256, 5)_{3}$ code, if one choose code with $(9, 256, 3)_{3}$, there is decrease in code rate. Hence there is open challenge for the construction of codes with optimal code rate for DNA based archival storage medium. Thus one can conclude that code of family $(11, 256, 5)_{3}$ is optimal code at present.
 \begin{table}[h]
\caption{Codewords from subcode of Ternary Golay Code i.e. $[11,6,5]_{3}$ assigned to 256 ASCII values is given in the table.}
\begin{center}
\begin{tabular}{|l|l|l|l|l|l|l|l|l|l|l|l|}
\hline
ASCII  &Golay codes & Wei- &ASCII & Golay codes& Wei-& ASCII & Golay codes&Wei-&ASCII &Golay codes & Wei-\\ 
Values&             & ght  &values&            &  ght  & values&            & ght   &values&       &ght\\ \hline\hline
86 & 00002111202 & 6&170 & 00001222101 &6&127& 00020220222 &6& 253&00022001121&6 \\ \hline
52 &00021112020 & 6&138 & 00010110111&6&41& 00012221010 &6&86&00011002212&6\\ \hline
42&00201010122 &6& 100 & 00200121021 &6&44&00202202220 &6& 250&00221200011&6\\ \hline
132 &00220011210 &6& 161 & 00222122112&9&98& 00211120200&6 & 8&00210201102&6\\ \hline
34 &00212012001&6& 10 & 00102020211 &6&149&00101101110&6 &87&00100212012&6 \\ \hline
21 &00122210100 &6& 74 & 00121021002 &6&36& 00120102201 &6& 69&00112100022&6\\ \hline
177 &00111211221 &9&20 & 00110022120&6&213& 02012212122&9&163&02011020021&6\\ \hline
229 &02010101220 &6& 255 & 02002102011 &6&197& 02001210210 &6& 133&02000021112&6\\ \hline
252 &02022022200 &6& 26 & 02021100102 &6&173&02020211001 &6&151&02210222211&9\\ \hline
82 &02212000110 &6& 75 & 02211111012&9 &37& 02200112100 &6&166&02202220002&6\\ \hline
191 &02201001201 &6&88& 02220002022 &6& 63& 02222110221&9& 68&02221221120 &9\\ \hline
150 &02111202000 &6&76 & 02110010202&5&4& 02112121101 &9& 154&02101122222 &9\\ \hline
234&02100200121 &6&22 & 02102011020&6 &162&02121012111&9& 105&02120120010&6\\ \hline
102 &02122201212&9 &171& 01021121211 &9&104& 01020202110&6 &169&01022010012&6\\ \hline
196& 01011011100 &6&208&01010122002&6& 84&01012200201&6&130&01001201022&6\\ \hline
146 &01000012221 &6& 72&01002120120&6 &16& 01222101000&6 & 66&01221212202&9\\ \hline
24 &01220020101 &6&106& 01212021222 &9& 223& 01211102121&9 & 58&01210210020&6\\ \hline
137 &01202211111 &9&73& 01201022010&6 & 101& 01200100212&6 & 168&01120111122&9\\ \hline
181&01122222021 &9&175&01121000220 &6& 251& 01110001011&6& 40&01112112210&9\\ \hline
140 &01111220112 &9& 17& 01100221200&6 & 83& 01102002102 &6&254 &01101110001&6\\ \hline
240&20121202122 &9&214& 20120010021&6 & 53& 20122121220 &9&202&20111122011 &9\\ \hline
25 &20110200210 &6&18&20112011112&9 & 247& 20101012200 &6& 174&20100120102&6\\ \hline
112 &20102201001 &6&89&20022212211&9 &210& 20021020110 &6& 217&20020101012&6\\ \hline
248 &20012102100 &6& 194&20011210002 &6&182& 20010021201&6&80&20002022022&6\\ \hline
79 &20001100221 &6& 195 &20000211120&6 &12& 20220222000 &6& 209&20222000202&6\\ \hline
165 &20221111101 &9&245&20210112222&9 &2& 20212220121 &9& 81&20211001020&6\\ \hline
38 &20200002111 &6& 141&20202110010 &6&211& 20201221212&9 & 239&22100111211&9\\ \hline
95 &22102222110 &9&43& 22101000012 &6&224&22120001100&6&203&22122112002&9\\ \hline
145 &22121220201&9 & 147&22110221022 &9&19& 22112002221 &9& 50&22111110120&9\\ \hline
136 &22001121000 &6& 107& 22000202202&6& 134& 22002010101&6 &109 &22021011222&9\\ \hline
153 &22020122121& 9&148& 22022200020 &6& 205& 22011201111 &9&212 &22010012010&6\\ \hline
54 &22012120212 &9& 241&22202101122&9& 156&22201212021&9 &115 &22200020220&6\\ \hline
116& 22222021011& 9& 78&22221102210 &9 & 67 & 22220210112 & 9 & 70 & 22212211200 &9 \\ \hline
178 &22211022102 &9& 159& 22210100001&6 & 142& 21112020000 &6&92 &21111101202&9\\ \hline
48&21110212101 &9& 90& 21102210222 &9& 218& 21101021121 &9&126 &21100102020&6\\ \hline
39 &21122100111 &9&219& 21121211010&9& 167&21120022212 &9&114 &21010000122&6\\ \hline
172&21012111021 &9& 14& 21011222220 &9& 120& 21000220011&6 &139 &21002001210&6\\ \hline
160&21001112112 &9& 33&21020110200 &6& 179& 21022221102&9&117&21021002001&6\\ \hline
225 &21211010211 &9&129&21210121110 &9& 183&21212202012 &9& 230&21201200100&6\\ \hline
35 &21200011002 &6&93&21202122201 &9& 6& 21221120022&9 &32&21220201221&9\\ \hline
56&21222012120 &9&158&10212101211 &9& 185& 10211212110&9 &47&10210020012&6\\ \hline
143 &10202021100 &6&123&10201102002 &6& 204& 10200210201 &6& 242&10222211022&9\\ \hline
111 &10221022221 &9&103&10220100120&6& 108& 10110111000&6& 9&10112222202&9\\ \hline
65 &10111000101 &6&249&10100001222 &6& 13&10102112121 &9& 180&10101220020&6\\ \hline
226&10120221111 &9&144&10122002010&6 & 15&10121110212 &9& 57&10011121122&9\\ \hline
128 &10010202021 &6&135&10012010220 &6& 243&10001011011&6 & 190&10000122210&6\\ \hline
207 &10002200112 &6&77&10021201200 &6& 45&10020012102 &6& 91&10022120001&6\\ \hline
192 &12221010000 &6&186&12220121202 &9& 216&12222202101&9 & 97&12211200222&9\\ \hline
118 &12210011121 &9&246&12212122020&9& 215&12201120111 &9& 51&12200201010&6\\ \hline
206 &12202012212 &9&184&12122020122 &9& 227&12121101021 &9& 233&12120212220&9\\ \hline
237 &12112210011&9 &188&12111021210 &9& 113&12110102112&9 & 49&12102100200&6\\ \hline
201 &12101211102&9& 155& 12100022001&6 & 222&12020000211&6&231 &12022111110&9\\ \hline
5 &12021222012 &9& 27& 12010220100 &6& 131& 12012001002 &6&164 &12011112201&9\\ \hline
3 &12000110022&6& 46& 12002221221&9 & 119& 12001002120&6 &28 &11200222122&9\\ \hline
176 &11202000021 &6 &23& 11201111220&9 &64&11220112011&9 &157 &11222220210&9\\ \hline
187&11221001112& 9&244& 11210002200&6 & 238& 11212110102 &9&96 &11211221001&9\\ \hline
235&11101202211 &9& 60&11100010110&6& 1&11102121012 &9&110 &11121122100&9\\ \hline
200 &11120200002&6& 221& 11122011201&9& 99&11111012022 &9&31 &11110120221&9\\ \hline
198 &11112201120 &9& 193&11002212000& 6&125&11001020202&6 &124 &11000101101&6\\ \hline
152 &11022102222& 9&122&11021210121&9& 71&11020021020&6 &94 &11012022111&9\\ \hline
220 &11011100010 &6& 29&11010211212&9& 199&00000201211 &5&61 &00000102122&5\\ \hline
11 &00002012110 &5&228&00002210021&5 & 62&00001021220 &5& 55&00001120012&5\\ \hline
121 &00020121100&5&7&00020022011 &5&30&00022100210 &5& 232&00022202002&5\\ \hline
189 &00021010201 &5&59&00010212200 &5& 236&00010011022&5 & 0&00000000000&0\\ \hline
\end{tabular}
\end{center}
\label{DNA golay code}
\end{table} 

\section{Conclusions}
In this paper, we introduced non linear family of DNA codes  with different parameters to develop efficient DNA based archival information storage system. Using non linear ternary Golay code (DNA Golay Code), we are able to correct 2 flip errors per block of 11 nucleotides along with significant improvement in storage capacity. Using this approach, we obtained significant difference in data storage density and length of DNA required when compared to Goldman's et al  approach. We also developed a software DNA cloud 2.0 which implements our encoding schemes available at http://www.guptalab.org/dnacloud/. We hope our codes generate interest in coding theory community to construct better coding schemes for this upcoming area. 
\bibliographystyle{IEEEtran} 
\bibliography{errorcorrection}
%
%


\end{document}